\documentclass[a4paper,autoref]{lipics-v2021}

\nolinenumbers

\usepackage{amsmath,amssymb,amsfonts,latexsym}
\usepackage{color,graphicx}
\usepackage{url} 
\usepackage{fancybox}
\usepackage{algorithm}
\usepackage[noend]{algpseudocode}
\usepackage{mathtools,thmtools}
\usepackage{enumerate,xspace}
\usepackage{siunitx} 
\hideLIPIcs

\newcommand{\eps}{\varepsilon}
\renewcommand{\epsilon}{\eps}
\newcommand{\etal}{\emph{et al.}\xspace}

\theoremstyle{plain}

\newenvironment{myquote}%
  {\list{}{\leftmargin=4mm\rightmargin=4mm}\item[]}%
  {\endlist}
\newenvironment{claiminproof}{\begin{myquote}\noindent\emph{Claim.}}{\end{myquote}}
\newenvironment{proofinproof}{\begin{myquote}\noindent\emph{Proof.}}{\hfill $\lhd$ \end{myquote}}

\newcommand{\factinproof}[2]{\begin{myquote}\noindent\emph{Fact~#1.} #2 \end{myquote}}

\newcommand{\A}{\ensuremath{\mathcal{A}}}
\newcommand{\B}{\ensuremath{\mathcal{B}}}
\newcommand{\C}{\ensuremath{\mathcal{C}}}
\newcommand{\D}{\ensuremath{\mathcal{D}}}

\newcommand{\F}{\ensuremath{\mathcal{F}}}
\newcommand{\G}{\ensuremath{\mathcal{G}}}
\newcommand{\ig}{\G^{\times}}               
\newcommand{\cH}{\ensuremath{\mathcal{H}}}
\newcommand{\cL}{\ensuremath{\mathcal{L}}}

\newcommand{\Q}{\ensuremath{\mathcal{Q}}}
\newcommand{\R}{\ensuremath{\mathcal{R}}}

\newcommand{\U}{\ensuremath{\mathcal{U}}}
\newcommand{\V}{\ensuremath{\mathcal{V}}}

\newcommand{\Integers}{{\mathbb{Z}}}
\newcommand{\Ints}{\Integers}
\newcommand{\REAL}{\ensuremath{\mathbb{R}}}
\newcommand{\Reals}{\REAL}

\renewcommand{\leq}{\leqslant}
\renewcommand{\geq}{\geqslant}
\newcommand{\bd}{\partial}

\DeclareMathOperator{\polylog}{polylog}

\DeclareMathOperator{\radius}{radius}
\DeclareMathOperator{\dist}{dist}
\DeclareMathOperator{\size}{size}
\DeclareMathOperator{\sibl}{sibl}

\DeclareMathOperator{\ply}{ply}

\newcommand{\tree}{\mathcal{T}}

\newcommand{\NE}{{\sc ne}}
\newcommand{\SE}{{\sc se}}
\newcommand{\SW}{{\sc sw}}
\newcommand{\NW}{{\sc nw}}

\newcommand{\vd}{\mathrm{Vor}}
\newcommand{\ue}{\mathrm{Env}}    
\newcommand{\dsource}{D_{\mathrm{src}}}
\newcommand{\sptree}{\tree_{\mathrm{sp}}}
\newcommand{\Dcand}{\D_{\mathrm{cand}}}
\newcommand{\lh}{\mbox{\sc lh}}    

\newcommand{\bit}{{\sc Bichromatic Intersection Testing}\xspace}

\def\DEF#1{\emph{#1}}    %

\title{An $O(n\log n)$ Algorithm for Single-Source Shortest Paths in Disk Graphs}

\author{Mark de Berg}{Department of Mathematics and Computer Science, TU Eindhoven, the Netherlands}{M.T.d.Berg@tue.nl}{0000-0001-5770-3784}{Supported by the Dutch Research Council (NWO) through    Gravitation-grant NETWORKS-024.002.003.}

\author{Sergio Cabello}{Faculty of Mathematics and Physics, University of Ljubljana, Ljubljana, Slovenia \and Institute~of~Mathematics, Physics and Mechanics, Ljubljana, Slovenia}{sergio.cabello@fmf.uni-lj.si}{0000-0002-3183-4126}{Funded in part by the Slovenian Research and Innovation Agency (P1-0297, N1-0218, N1-0285).
Funded in part by the European Union (ERC, KARST, project number 101071836). Views and opinions expressed are however those of the authors only and do not necessarily reflect those of the European Union or the European Research Council. Neither the European Union nor the granting authority can be held responsible for them.}

\authorrunning{M.~de Berg and S.~Cabello} 
\Copyright{Mark de Berg and Sergio Cabello}



\begin{document}

\maketitle

\begin{abstract}
We prove that the single-source shortest-path problem on 
disk graphs can be solved in $O(n\log n)$ time, and that it
can be solved on intersection graphs of fat triangles in $O(n\log^2 n)$ time.
\end{abstract}

\section{Introduction}
\label{sec:intro}
Finding shortest paths in graphs is a classic topic
covered in all basic courses and textbooks on algorithms. 
In the most basic setting, which is the setting we consider here, the graph
is unweighted and the length of a path refers 
to the number of edges in the path. In the single-source shortest-path (SSSP) problem, 
the task is to compute shortest paths from a given vertex, 
the \emph{source}, to all other vertices of the graph. The solution 
is represented using a so-called \emph{shortest-path tree}.

The SSSP problem for a graph with $n$ vertices 
and $m$ edges can be solved in $O(n+m)$ time by breadth-first search, which
is optimal if the graph is given as a collection of vertices and edges. 
In this paper we will consider the SSSP problem for implicitly defined graphs.
In particular, we consider (planar) \DEF{intersection graphs}: graphs whose node set 
corresponds to a set~$\D$ of $n$~objects in the plane and that contain an edge~$(D,D')$ 
between two distinct objects $D,D'\in\D$ iff $D$ and $D'$ intersect. 
If the objects in $\D$ are disks then the intersection graph, which we denote by~$\ig[\D]$,
is called a \emph{disk graph}---see Figure~\ref{fig:intersection-graph}---and if all disks have
unit radius then it is called a \emph{unit-disk graph}. 
Disk graphs are arguably the most popular and widely studied intersection graphs. 
One can solve the SSSP problem on intersection graphs
by first constructing~$\ig[\D]$ explicitly and then running breadth-first search.
In the worst case, however, this requires $\Omega(n^2)$ time. This raises
the question: is it possible to solve the SSSP problem on intersection
graphs in subquadratic time in the worst case?
\begin{figure}[b]
\begin{center}
\includegraphics{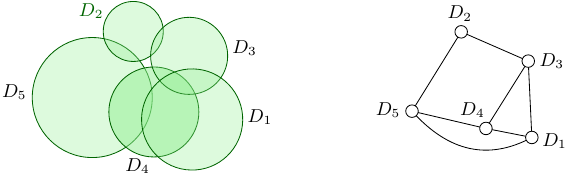}
\end{center}
\caption{A set~$\D$ of five disks (left) and their intersection graph (right).}
\label{fig:intersection-graph}
\end{figure}

Given the basic nature of the SSSP problem and the prominence of (unit-)disk graphs,
it is not surprising that this question has already been considered for such graphs. 
For unit disks, Roditty and Segal~\cite{RodittyS11} noticed that the dynamic data structure 
of Chan~\cite{Chan10} for nearest-neighbor queries can be used
to solve the problem in $O(n\polylog n)$ expected time.
Cabello and Jej{\v c}i{\v c}~\cite{CabelloJ15} gave an $O(n \log n)$ algorithm
and showed that this is asymptotically optimal in the algebraic decision-tree model. 
They remarked that Efrat noted that the semi-dynamic 
data structure of Efrat, Itai and Katz~\cite{EfratIK01} 
also gives an algorithm with running time $O(n \log n)$.
Finally, Chan and Skrepetos~\cite{ChanS16} provided a simpler $O(n \log n)$ algorithm.
The above algorithms either use Delaunay triangulations~\cite{CabelloJ15} 
or fixed-resolution grids~\cite{ChanS16,EfratIK01}; these approaches are not
applicable for disks of arbitrarily different sizes.

For arbitrary disks graphs, Kaplan~\etal~\cite{KaplanMRSS20} and Liu~\cite{doi:10.1137/20M1388371}
presented algorithms running in~$O(n\log^4 n)$ time.  This was recently improved
by Klost~\cite{Klost23}, who presented a general framework for solving SSSP 
problems in intersection graphs. She used the framework
to obtain an~$O(n\log^2 n)$ algorithm for the SSSP problem on disk graphs, and 
to obtain an $O(n\log n)$ algorithm for intersection graphs of axis-aligned squares.

\subparagraph{Our results.}
Our main result is an $O(n\log n)$ algorithm for the SSSP problem
in disk graphs. We obtain this result using a version of the framework proposed by Klost~\cite{Klost23};
our specialized framework and its relation to Klost's framework are
discussed in detail in Section~\ref{sec:framework}. A core ingredient in
the framework is a \emph{clique-based contraction} of the intersection graph~$\ig[\D]$:
a graph obtained from~$\ig[\D]$ by contracting certain cliques to single nodes.
To compute such a clique-based contraction
in $O(n\log n)$ time, we combine shifted quadtrees~\cite{C-PTAS-fat}
and skip quadtrees~\cite{EppsteinGS08}---powerful 
and beautiful concepts that we believe deserve more attention---with standard
techniques for stabbing fat objects.

Our algorithm for computing a clique-based contraction not only applies
to disk graphs but also to intersection graphs of \emph{fat triangles},
that is, triangles all of whose internal angles are lower bounded by
a fixed constant~$\alpha>0$~\cite{AronovBES14}. By combining this
with a novel intersecting-detection data structure for fat triangles,
we obtain an $O(n\log^2 n)$ algorithm for the SSSP problem 
on intersection graphs of fat triangles.
We are not aware of previous work on the SSSP
problem for fat triangles. 
The all-pairs shortest-path (APSP) problem, however, has been considered for
fat triangles, by Chan and Skrepetos~\cite{ChanS19}.
The algorithm they present runs in $O(n^2 \log^4 n)$ time, under the condition
that the fat triangles have roughly the same size.
(They also present an $O(n^2\log n)$ algorithm for the APSP problem for arbitrary disk graphs, 
and an $O(n^2)$ algorithm for unit-disk graphs~\cite{ChanS16}.)
Running our new SSSP algorithm $n$ times, 
once for each input triangle as the source, we obtain an APSP algorithm for 
arbitrarily-sized fat triangles
that runs in $O(n^2\log^2 n)$ time. This is faster and more general 
than the algorithm of Chan and Skrepetos.
(The improvement for the APSP problem is mainly
because of our new intersection-detection data structure
for fat triangles; using this data structures in
the existing algorithm~\cite{ChanS19} would also give an $O(n^2\log^2 n)$ algorithm.)

\subparagraph{On the model of computation.}
We use the real-RAM model---this is the standard model in
computational geometry, which allows us for example
to check in $O(1)$ time if two disks intersect---extended with
an additional operation that allows us to compute a compressed
quadtree on a set of $n$ points in~$O(n\log n)$ time.
This extension is common when working with quadtrees;
see the book by Har-Peled~\cite{HarPeledBook}.

\section{The framework}
\label{sec:framework}
Our framework for computing a shortest-path tree on an intersection 
graph~$\ig[\D]$ is an instantiation of the framework of 
Klost~\cite{Klost23}. We describe it
in detail to keep the paper self-contained and make the required subroutines explicit; 
at the end of this section we
discuss the correspondence of our approach to that of Klost.
For convenience, we will from now on not distinguish between the objects in the set $\D$
and the corresponding nodes in~$\ig[\D]$.

Let $\dsource\in\D$ be the given source node, and let $\sptree$ be the shortest-path tree 
of $\ig[\D]$ that we want to compute for the given source. 
For an object $D\in \D$, let $\dist[D]$ be the distance from
$\dsource$ to $D$ in $\ig[\D]$, and let $L_\ell := \{D\in \D : \dist[\D]=\ell\}$. 
Thus, $L_\ell$ is the set of nodes at level~$\ell$ in~$\sptree$. 
We follow the natural approach of computing $\sptree$ level by level. 
To be able to go from one level to the next, we need a subroutine for
the following problem.
\begin{quotation}
\noindent \bit. Given a set~$B$ of $n_B$ blue objects and a set~$R$
of $n_R$ red objects, report the blue objects that intersect at least one red object,
and for each reported blue object, report a witness (a red object intersecting it).
\end{quotation}
A subroutine \emph{BIT-Subroutine} that solves \bit allows us to compute $L_{\ell+1}$ by setting
$R := L_{\ell}$ and $B := \Dcand$, where $\Dcand$ is a set of 
\emph{candidate objects} that should contain all objects from~$L_{\ell+1}$ and
no objects from $L_{\leq \ell}$, where
$L_{\leq \ell} := \bigcup_{i=0}^{\ell} L_i$.
We could simply set $\Dcand := \D \setminus L_{\leq \ell}$,
but this will not be efficient; we must keep the total size of the
candidates sets over all levels~$\ell$ under control. We do this using
a so-called clique-based contraction, as defined next.
\smallskip

Let $\G=(V,E_{\G})$ be an (undirected) graph. 
We say that a graph $\cH=(\C,E_{\cH})$ is a \DEF{clique-based contraction} of $\G$ if
\begin{itemize}
\item each node $C\in\C$ corresponds to a clique in $\G$ and the cliques in $\C$ form
      a partition\footnote{In our application, it would also be sufficient to require 
      that each node in $V$ appears in $O(1)$ cliques, but we do not need this extra flexibility.} of the nodes in $V$;
\item there is an edge $(C,C')\in E_{\cH}$ iff there are nodes $v\in C$ and $v'\in C'$ 
      such that $(v,v')\in E_{\G}$. 
\end{itemize}
\smallskip
Note that for each edge $(v,v')\in E_{\G}$ there exists a clique $C\in \C$ that
contains both $v$ and $v'$, or there are cliques $C\ni v$ and $C'\ni v'$ such that
$(C,C')\in E_{\cH}$.
\medskip

A clique-based contraction $\cH$ of the intersection graph~$\ig[\D]$ helps us to find good candidate sets.
To see this, let $\C_{\ell}\subset \C$ denote the set of cliques\footnote{With a slight abuse of terminology,
we not only use the term \emph{clique} to refer a complete graph, but also to refer to
a set of geometric objects whose intersection graph is a clique.} that contain
at least one object from $L_{\ell}$. Then any object $D\in L_{\ell+1}$ 
must be part of a clique $C$ that is either in $\C_{\ell}$ itself or that
is a neighbor in $\cH$ of a clique in~$\C_{\ell}$. In other words, if $D\in L_{\ell+1}$  then $D\in C$ for some clique $C\in N_{\cH}[\C_{\ell}]$,
where $N_{\cH}[\C_{\ell}]$ denotes the closed neighborhood of $\C_\ell$ in~$\cH$.
Moreover, the objects in a clique are in at most two levels in the 
shortest-path three~$\sptree$, which helps to control the total size of the candidate sets.
The following pseudocode describes our framework in detail.
It does not explicitly construct the shortest-path tree itself
(it only constructs the levels), but this is easy to do using the witnesses reported
by \emph{BIT-Subroutine}.
\begin{algorithm}[H] 
\caption{\emph{SSSP-for-Geometric-Intersection-Graphs}($\D,\dsource$)}
\begin{algorithmic}[1]
\State Construct a clique-based contraction $\cH=(\C,E_\cH)$ of $\ig[\D]$ \label{step:contract}
\State $L_0 \gets \{ \dsource \}$; \ \ $\C_0 \gets \{ \mbox{the clique $C\in \C$ that contains~$\dsource$} \}$ \hfill $\rhd$ $\dsource$ is the source 
\State Label $\dsource$ as \emph{ready}; \ \ $\ell\gets 0$; \ \  $\mbox{\emph{done}}\gets\mbox{{\sc false}}$. \label{step:init}
\While {not \emph{done}} 
    \State $\Dcand \gets \{ D \in  \D : \mbox{$D\in C$ for a clique $C\in N_{\cH}[\C_\ell]$ and $D$ is not \emph{ready}} \}$   \label{step:cand}
    \State $L_{\ell+1} \gets \mbox{\emph{BIT-Subroutine}}(\Dcand,L_{\ell})$ \label{step:bit}
    \If{$L_{\ell+1}=\emptyset$} \label{step:if}
        \State  $L_{\infty} \gets \D \setminus L_{\leq \ell}$; \ \ 
                $\mbox{\emph{done}}\gets\mbox{{\sc true}}$ \hfill $\rhd$ nodes in $L_{\infty}$ are unreachable from $\dsource$ 
    \Else 
    \State $\C_{\ell+1} \gets \{ \mbox{$C\in \C$ : $C$ contains an object $D\in L_{\ell+1}$} \} $
    \State Label all objects in $L_{\ell+1}$ as \emph{ready}; \ \  $\ell \gets \ell+1$ \label{step:ready}
    \EndIf
\EndWhile
\end{algorithmic}
\end{algorithm}
\noindent We obtain the following theorem.
\begin{theorem} \label{thm:framework}
Let $\D$ be a set of $n$ constant-complexity objects in the plane. Suppose that
\begin{enumerate}[(i)]
\item we can compute a clique-based contraction for $\ig[\D]$ in $T_{\mathrm{ccg}}(n)$ time,
\item for any two subsets $B,R\subset \D$ we can solve \bit in $T_{\mathrm{bit}}(n_B,n_R)$ time, where $n_B := |B|$ and $n_R := |R|$.
\end{enumerate}
Then we can compute a shortest-path tree in $\ig[\D]$ for a given source node $\dsource\in \D$ in 
$O(T_{\mathrm{ccg}}(n) + T_{\mathrm{bit}}(n,n))$ time.
\end{theorem}
\begin{proof}
It is straightforward to prove by induction on~$\ell$ that our algorithm correctly 
computes the levels~$L_\ell$ of the shortest-path tree. 
To prove the bound on the running time, observe that Steps~\ref{step:contract}--\ref{step:init}
take $O(T_{\mathrm{ccg}}(n))$ time. It remains to bound the runtime of the while-loop.

For a clique $C\in \C$, let $\dist[C] := \min_{D\in C} \dist[D]$. 
Note that for objects $D,D'$ in cliques $C,C'\in \C$ that are neighbors in~$\cH$, 
we have $|\dist[D]-\dist[D']\,|\leq 3$.
Hence, $D$ can only be in the candidate set $\Dcand$ in iterations of the while-loop
where $\dist[D]-3\leq \ell\leq \dist[C]-1$. Thus, if we denote the size 
of $\Dcand$ in iteration~$\ell$ by $n_{\ell}$, then the total time needed
for Step~\ref{step:bit} over all iterations 
is $\sum_{\ell} T_{\mathrm{bit}}(n_\ell,|L_\ell|)$, 
where $\sum_{\ell} n_\ell \leq 3n$ and $\sum_{\ell} |L_\ell| \leq n$.
Since $T(n_B,n_R)$ is at least linear in $n_B$ and $n_R$
(and at most quadratic), we have 
$\sum_{\ell} T_{\mathrm{bit}}(n_\ell,|L_\ell|)=O(T_{\mathrm{bit}}(n,n))$.
This also bounds the total runtime of Steps~\ref{step:if}--\ref{step:ready} 
over all iterations.

To bound the total time for Step~\ref{step:cand}, note that a clique $C$ is only
considered in iterations where $\dist[C]-2 \leq \ell \leq \dist[C]+2$. 
This implies that the total time to find the relevant cliques in Step~\ref{step:cand}
is $O(|E_\cH|)$. The total time to inspect those cliques in order to find the candidate
sets, is $O(\sum_{C\in \C} |C|)$. Thus, the total time for Step~\ref{step:cand}
over all iterations is $O(|E_\cH| + \sum_{C\in \C} |C|)$, 
which can be bounded by $O(T_{\mathrm{ccg}}(n))$.
\end{proof}
\noindent \emph{Remark: relation to the framework of Klost.}
The framework presented above is an instantiation of the framework of 
Klost~\cite{Klost23}: she also computes the shortest-path tree
level by level as described above---this natural approach was also taken by earlier papers
on the SSSP problem on (unit-)disk 
graphs~\cite{CabelloJ15,ChanS16,ChanS19,Klost23}---and she also keeps the size
of the candidate sets under control using an auxiliary graph. Klost uses
a so-called \emph{shortcut graph} $\G_{\mathrm sc}$ as auxiliary graph. 
The nodes in $\G_{\mathrm sc}$ correspond to the objects in~$\D$,
and the edge set $E_{\mathrm sc}$
is a superset of the edge set of $\ig[\D]$ such that for any edge
$(D,D')\in E_{\mathrm sc}$ we have that the distance between $D$ and $D'$ in 
$\ig[\D]$ is~$O(1)$. As $E_{\mathrm sc}$ is a superset of the edge set of~$\ig[\D]$,
the size of $E_{\mathrm sc}$ can be quadratic, so $\G_{\mathrm sc}$ needs
to be constructed implicitly. Our clique-based contraction~$\cH=(\C,E_cH)$ provides
such an implicit shortcut graph, by defining 
$E_{\mathrm sc} := \{\,(D,D') :$ $D$ and $D'$ are part of the same clique in $\C$ 
or of neighboring cliques$\,\}$.

\section{The algorithm for disks}
\label{sec:disks}
In this section we implement the framework presented above for the case where the input $\D$ 
is a set of $n$ disks in the plane. We first explain
how to create the cliques in our clique-based contraction~$\cH$ of $\ig[\D]$, and how
to compute the edge set of $\cH$. Using known results on \bit
for disks we then obtain our final result.

\subparagraph{Finding the cliques.}
Our algorithm to create the set of cliques is related to a recent 
construction by Chan and Huang~\cite[Section~3]{ChanH23} 
of 3-hop spanners for disk graphs; we comment more on the similarities and differences later. 
The construction is based on so-called shifted quadtrees, introduced by Chan~\cite{C-PTAS-fat},
which we describe next. 

We start by defining hierarchical grids, a concept closely related to quadtrees;
the connection will be made clear below.
Let $o=(o_x,o_y)$ be any point in the plane.
A \DEF{cell} in the hierarchical grid centered at $o$ is any square of the form
\[
\sigma(o,\ell,i,j):=~ 
	\big[ o_x+i \cdot 2^{\ell}, o_x+ (i+1) \cdot 2^{\ell} \big) \times 
	\big[ o_y+j \cdot 2^{\ell}, o_y+(j+1) \cdot 2^{\ell} \big)
\]
for integers $i,j,\ell\in \Ints$. Note that $\sigma(o,\ell,i,j)$
is a translation of the square $[0,2^{\ell})\times [0,2^{\ell})$ by
the vector $o+(i \cdot 2^{\ell},j \cdot 2^{\ell})$, and that
the set $\Sigma(o,\ell): = \{\sigma(o,\ell,i,j): i,j \in \Ints\}$ 
forms a regular grid. Moreover, the grid~$\Sigma(o,\ell-1)$ is a refinement
of the grid~$\Sigma(o,\ell)$, obtained by partitioning each cell
of $\Sigma(o,\ell)$ into four quadrants.
We define the \DEF{hierarchical grid centered at} $o$, denoted by $\Gamma(o)$,
to be the (infinite) collection~$\{ \Sigma(o,\ell) :\ell\in\Ints\}$ of nested grids.

Define the \DEF{size} of any object~$D$, denoted by $\size(D)$, to
be the side length of a smallest enclosing square of~$D$.
Thus $\size(\sigma)=2^{\ell}$ for any cell $\sigma\in\Sigma(o,\ell)$.
We say that an object~$D$ is \DEF{$c$-aligned} with a hierarchical
grid~$\Gamma$, for a given constant~$c>0$, if there exists a cell~$\sigma$ 
in~$\Gamma$ such that $D\subset \sigma$ and $\size(\sigma)\leq c\cdot \size(D)$. 
Note that whether or not $D$ is $c$-aligned with $\Gamma$ only depends 
on the choice of the center~$o$ of the hierarchical grid. The definition of 
being $c$-aligned trivially implies the following. 
\begin{observation} \label{obs:c-aligned}
Let $D$ be an object that is $c$-aligned with a hierarchical grid $\Gamma$, 
and let $\sigma$ be a cell of $\Gamma$ such that $D$ intersects the 
boundary~$\bd \sigma$ of~$\sigma$.
Then $\size(D) \geq (1/c)\cdot \size(\sigma)$.
\end{observation}
Let $\D$ be a finite set of objects.
In general, it is impossible to choose the center of a hierarchical grid $\Gamma$ 
such that each object~$D\in \D$ is $c$-aligned with~$\Gamma$. However, Chan~\cite{C-PTAS-fat}
has shown that, surprisingly, it \emph{is} possible to pick a small number
of different centers such that each object~$D\in\D$ is $c$-aligned with 
at least one of the resulting hierarchical grids, for a suitable constant~$c$. 
Even more surprisingly, we can select these centers independently of~$\D$.
The following lemma follows directly from Lemma~3.2 in Chan's paper.\footnote{Chan uses
the term \emph{shifted quadtree} in his lemma. To stress the fact
that the shifts are independent of any input set on which one might build a quadtree,
we prefer to use the term \emph{hierarchical grid}.}
\begin{lemma}[\cite{C-PTAS-fat}] \label{lem:shifted-quadtree}
There is a collection~$\Xi$ of three hierarchical grids, each centered at a different point, 
with the following property: for any object $D$ contained in the square $[0,1)\times [0,1)$
there is a hierarchical grid $\Gamma\in \Xi$ such that $D$ is 6-aligned with~$\Gamma$.
\end{lemma}
Let~$P$ be a set of points, let~$\Gamma$ be a hierarchical grid centered
at a given center~$o$, and let~$\sigma_0$ be a cell of~$\Gamma$ that contains~$P$. 
The we can construct a \DEF{quadtree} on $P$---see Figure~\ref{fig:quadtree}(i)---whose 
root node corresponds to~$\sigma_0$.
Any node in this quadtree corresponds to a cell of $\Gamma$.
Let $\Q=\Q(\sigma_0,P)$ be the corresponding \DEF{compressed quadtree}~\cite{HarPeledBook}.
Each internal node~$\nu$ of $\Q$ corresponds to a 
cell~$\sigma_{\nu}$ of $\Gamma$ and each leaf node $\mu$ corresponds to
a region~$\sigma_{\mu}$ that is either a cell of $\Gamma$ or a donut cell;
see Figure~\ref{fig:quadtree}(ii).
(A \DEF{donut cell} is the difference $\sigma_1\setminus\sigma_2$ of
two cells~$\sigma_1$ and~$\sigma_2$ of~$\Gamma$ with $\sigma_2\subset\sigma_1$.)
The \DEF{subdivision} defined by the compressed quadtree~$\Q$ is the set of regions
corresponding to the leaves of~$\Q$; this subdivision is a decomposition of~$\sigma_0$.
\begin{figure}
\begin{center}
\includegraphics[width=\textwidth]{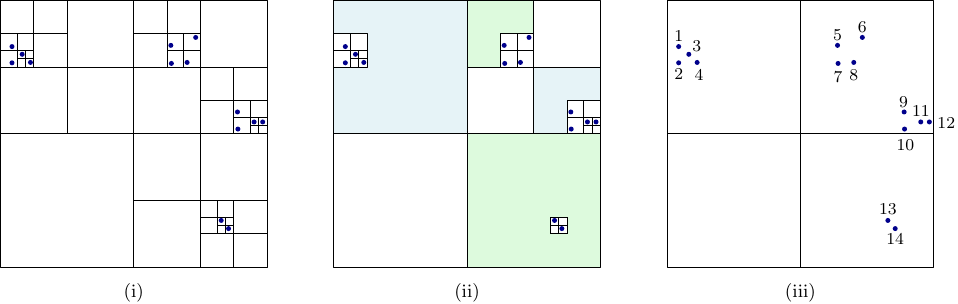}
\end{center}
\caption{(i)~A quadtree subdivision for a set $P$ of points.
(ii)~The compressed-quadtree subdivision for~$P$, 
with its donut cells marked. 
(iii)~The ordering of the points to construct a skip quadtree.} 
\label{fig:quadtree}
\end{figure}
Our algorithm to efficiently compute the cliques in the clique-based contraction combines 
the hierarchical grids of Lemma~\ref{lem:shifted-quadtree} with 
skip quadtrees~\cite{EppsteinGS08}, as explained next. 
\medskip

Skip quadtrees, defined by Eppstein, Goodrich, and Sun~\cite{EppsteinGS08}
are defined as follows. Let~$P$ be a set of $n$ points inside a given starting  
cell~$\sigma_0$ of a hierarchical grid~$\Gamma$. We start by sorting the 
points in~$P$ in Z-order, as follows.
First, order the points according to the quadrants from~$\sigma_0$ they fall in:
put the points in the \NW-quadrant first, then those in the \NE-quadrant,
then those in the \SW-quadrant, and finally those in the \SE-quadrant.
Next, recursively sort the points in each quadrant. 
Let $p_1,\ldots,p_n$ be the resulting sorted list
of points; see Figure~\ref{fig:quadtree}(iii).
We now create a sequence $P=P_0 \supset P_1 \supset \cdots \supset P_{t}$,
where $P_i$ is obtained from $P_{i-1}$ by deleting every other element (starting with the first element)
and $|P_t|=1$.
From the sequence $P_0 \supset P_1 \supset \cdots \supset P_{t}$ we construct
a sequence of compressed-quadtree \emph{subdivisions} $\Q_0,\Q_1,\ldots,\Q_{t}$, each based
on the same starting cell~$\sigma_0$. 
The sequence $\Q_0,\Q_1,\ldots,\Q_{t}$
has the following properties, as illustrated in Figure~\ref{fig:quadtree2}:
\begin{itemize}
\item $\Q_t$ is equal to the square~$\sigma_0$.
\item $\Q_i$ is a refinement of $\Q_{i+1}$ for all $0\leq i<t$, that is, 
      each region in $\Q_i$ is contained in a unique region in~$\Q_{i+1}$.
      Moreover, each region in $\Q_{i+1}$ contains $O(1)$ regions from~$\Q_{i}$.
\item Each region in $\Q_0$ contains at most one point from~$P$.
\end{itemize}
We can view $\Q_t,\ldots,\Q_1,\Q_{0}$ as the levels of a tree~$\tree$, 
where a node $\nu$ at level~$i$ corresponds to a region~$R_{\nu}$ of
the subdivision~$\Q_i$. In particular, the root of~$\tree$ corresponds
to the starting square~$\sigma_0$ and the leaves of~$\tree$ correspond to the regions
in~$\Q_0$ (which is the compressed-quadtree subdivision for the whole point set~$P$ with
root~$\sigma_0$). We 
call this tree structure\footnote{The skip quadtree defined by Eppstein, Goodrich, and Sun
is slightly more complicated, as it stores more information than just the
compressed-quadtree subdivisions. Since we do not need our structure to be
dynamic, the simple hierarchy described above suffices.}
a \DEF{skip quadtree} for~$P$ with root~$\sigma_0$. It 
can be constructed in $O(n \log n)$ time~\cite{EppsteinGS08}.
\begin{figure}
\begin{center}
\includegraphics{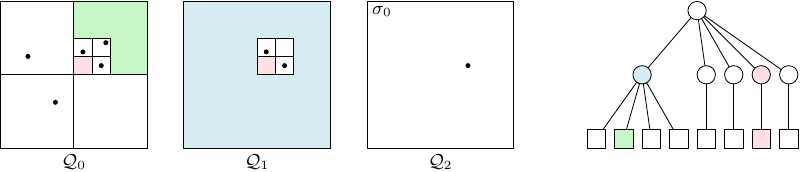}
\end{center}
\caption{Example of the three compressed-quadtree subdivisions for a set of five points, and the resulting skip quadtree. Colors indicate the correspondence between
regions and nodes.}
\label{fig:quadtree2}
\end{figure}
\medskip

Without loss of generality, we will assume henceforth that all the disks of $\D$ are fully
contained in the unit square $[0,1)\times [0,1)$.
Let $\Xi$ be the collection of three hierarchical grids given by 
Lemma~\ref{lem:shifted-quadtree}.
Each disk~$D$ of $\D$ is $6$-aligned with some hierarchical grid $\Gamma\in \Xi$,
and we can find such a hierarchical grid for $D$ in constant time.

Now consider a hierarchical grid $\Gamma$ from the collection $\Xi$. 
Let $\D(\Gamma)\subset \D$ be the set of disks that are 6-aligned with~$\Gamma$, 
where we put a disk that is 6-aligned with multiple grids  
into an arbitrary one of the corresponding sets $\D(\Gamma)$.
We create a set $\C(\Gamma)$ of cliques, as follows. 
\begin{enumerate}
\item \label{step1} Construct a skip quadtree~$\tree(\Gamma)$ on the set~$P(\Gamma)$
	  of centers of the disks in~$\D(\Gamma)$,
      where the starting square~$\sigma_0$ of the skip quadtree is a cell of~$\Gamma$ 
	  that contains all the disks in~$\D(\Gamma)$.
\item \label{step2} For each disk $D\in\D(\Gamma)$, follow a path in $\tree(\Gamma)$ consisting of 
      nodes~$\mu$ whose region~$R_{\mu}$ contains the center of~$D$, until a node~$\nu$
      is reached such that $D$ intersects~$\bd R_{\nu}$ or a leaf is reached.
      Assign $D$ to that node or leaf.
\item \label{step3} Let $\D_{\nu}\subset\D(\Gamma)$ be the set of disks assigned to~an 
      (internal or leaf) node~$\nu$ in $\tree(\Gamma)$. For each node $\nu$, partition $\D_{\nu}$ 
      into a set $\C_{\nu}$ of $O(1)$ cliques, as follows.  If~$\nu$ is a leaf and $R_{\nu}$ 
      contains a disk---note that a leaf can contain at most one disk---then 
	   we create a singleton clique for that disk. Now suppose $\nu$ is an internal node. 
      First, suppose $R_{\nu}$ 
      is a square region.
      Because the disks in $\D_{\nu}$ are 6-aligned with~$\Gamma$, we know
      that $\size(D)\geq \tfrac{1}{6}\cdot \size(R_{\nu})$ for all $D\in\D_{\nu}$
	  intersecting~$\bd R_{\nu}$.
      Hence, we can create a set of $O(1)$ points such that each disk $D$
      is stabbed by at least one of them, which we use to create~$\C_{\nu}$. 
      If $R_{\nu}$ is a donut we proceed similarly, where we treat the
      disks intersecting the boundary of the hole of the donut---this
      hole can be much smaller than the donut---separately
      from the disks that only intersect the outer boundary of the donut.
      Finally, set $\C(\Gamma) := \bigcup_{\nu\in \tree(\Gamma)} \C_{\nu}$.      
\end{enumerate}
Each clique that is generated by our algorithm is a so-called
\DEF{stabbed clique}, that is, for each clique $C\in\C(\Gamma)$ there is a point stabbing 
all disks in~$C$. We refer to the union of all the disks in~$C$ as a \DEF{flower},
which we denote by~$F(C)$, and we let $\F(\Gamma) := \{ F(C) : C\in \C(\Gamma)\}$
be the set of flowers corresponding to the cliques in~$\C(\Gamma)$. 
Define $\ply(S)$, the \DEF{ply} of a set $S$ of objects in the plane, to be the maximum, over
all points  $q\in \Reals^2$, of the number of objects in~$S$ containing~$q$. 
The following lemma states a property of our algorithm that will be crucial
to efficiently compute the edges in the clique-based contraction.
\begin{lemma}\label{lem:clique-construction}
The algorithm above constructs in $O(n\log n)$ time a partition of $\D(\Gamma)$ into a set $\C(\Gamma)$ of stabbed cliques such that $\ply(\F(\Gamma))=O(\log n)$.
\end{lemma}
\begin{proof}
Step~\ref{step1} of the algorithm, constructing the skip quadtree $\tree(\Gamma)$ 
on the set~$P(\Gamma)$ of the centers of the
at most~$n$ disks in~$\D(\Gamma)$, can be done in~$O(n\log n)$
time~\cite{EppsteinGS08}. The depth of the skip quadtree is
$O(\log n)$, and so Step~\ref{step2} takes $O(\log n)$ time per disk and
$O(n\log n)$ time in total. Step~\ref{step3} takes $\sum_{v\in \tree(\Gamma)} O(1+|\D_v|)$ time, which is $O(n)$ because each disk is assigned to exactly one node.
Hence, the total time for the algorithm is~$O(n\log n)$.

To bound~$\ply(\F(\Gamma))$, consider an arbitrary point~$q\in\Reals^2$.
Recall that for each node $\nu$ in $\tree(\Gamma)$ we created $O(1)$ stabbed 
cliques and, hence, $O(1)$ flowers. Thus, it suffices to prove the following claim:
at each level of the skip quadtree~$\tree(\Gamma)$ there are only $O(1)$
nodes $\nu$ such that $q\in D$ for some disk~$D$ assigned to~$\nu$.
To prove this claim, follow the search path of $q$ in $\tree(\Gamma)$, 
that is, the path consisting of nodes $\mu$ such that $q\in R_{\mu}$. 
Let $D\in \D$ be a disk that
contains~$q$, let $\nu$ be the node to which $D$ is assigned, and let $\mu$
be the parent of~$\nu$. Then $\mu$ must lie on the search path of~$q$.
Indeed, $D$ is fully contained in $R_{\mu}$---otherwise $D$ would have been 
assigned to~$\mu$---and so $q\in D\subset R_{\mu}$. Since any node in a skip quadtree
has $O(1)$ children, this implies the claim.
\end{proof}

\noindent \emph{Remark: relation to the spanner construction of Chan and Huang.}
In their algorithm for constructing a 3-hop spanner for~$\ig[\D]$, 
Chan and Huang~\cite{ChanH23} create a collection
of cliques in a similar way as we create our cliques. As they are only concerned
with the size of the spanner, however, they do not analyze the time needed to construct it.
The differences between our approach and theirs mainly stem from the necessity
to construct the clique-based contraction~$\cH$ is $O(n\log n)$ time,
as discussed next.

A first difference is that Chan and Huang use a centroid decomposition of a 
compressed quadtree instead of a skip quadtree.
This is not a crucial difference---possibly we could have used a centroid decomposition
as well, although working with a skip quadtree is more convenient.
Another difference is that they work with a collection~$\Xi$ of hierarchical grids---they
call them shifted quadtrees---such that,
for each pair $D,D'\in\D$, there is at least one hierarchical grid $\Gamma\in\Xi$ such that 
both $D$ and $D'$ are aligned with~$\Gamma$. This allows them to construct their 
spanner by applying the following procedure to each quadtree~$\Q$ defined by the hierarchical
grid~$\Gamma$: 
(i)~take a centroid node $\nu$ in a compressed quadtree~$\Q$ and
partition the set~$\D_{\nu}$ of disks intersecting $\bd \sigma_{\nu}$ into $O(1)$ cliques, 
(ii)~for each clique~$C$, add a star graph on that clique to the spanner,
(iii)~for each disk $D$ and each clique~$C$ intersected by~$D$, add an edge $(D,D')$ to the spanner, 
where $D'\in C$ is an arbitrary disk intersected by~$D$, and
(iv)~recursively construct 3-hop spanners for the disks inside and outside $\sigma_{\nu}$.
The crucial difference lies in step~(iii), where they add an edge from every disk~$D$
to a neighboring disk in each clique (if such a neighboring disk exists). It seems difficult to do this
in $O(n\log n)$ time in total over all recursive calls. 
We therefore proceed differently: we first recursively construct cliques on the disks not intersecting~$\bd \sigma_{\nu}$, and then we only add a single edge between
two cliques if the corresponding flowers intersect. Such an approach
would increase the spanning ratio if we were to use it construct a spanner,
but this is of no concern to us. The advantage of being more conservative
is that we can compute the edges in our clique-based contraction
in $O(n\log n)$ time, as explained next.

\subparagraph{Computing the edges in the clique-based contraction.}
Let $\C := \bigcup_{\Gamma\in\Xi} \C(\Gamma)$ be the set of cliques generated by applying 
the procedure above to each hierarchical grid $\Gamma$ in the collection~$\Xi$ given 
by Lemma~\ref{lem:shifted-quadtree}. Let $\F := \bigcup_{\Gamma\in \Xi} \F(\Gamma)$ 
be the corresponding flower set. Because of Lemma~\ref{lem:clique-construction},
the construction of $\C$ takes $O(n \log n)$ time and $\ply(\F)=O(\log n)$. 
To construct the edge set~$E_{\cH}$  
of the clique-based contraction~$\cH=(\C,E_{\cH})$, we must find all 
pairs~$F,F'\in \F$ that intersect. 
Next we show how to do this in $O(n\log n)$~time.

The algorithm below may report a pair of intersecting flowers multiple times,
but this is not a problem. In total it produces $O(n\log n)$ edges for~$E_{\cH}$,
counted with multiplicity. We can just work with the multigraph or 
we may clean the multigraph removing copies of edges to get the underlying simple graph;
this cleaning step takes linear time in the size of the 
multigraph~\cite[Exercise~22.1-4]{CLRS09}, which is $O(n\log n)$.

Note that two flowers $F,F'$ intersect iff the following holds: 
a boundary arc of $F$ intersects a boundary arc of $F'$, or $F\subset F'$, 
or $F'\subset F$. We will first concentrate on the intersection between
boundaries of the flowers, and then consider the containment between flowers.
\medskip

Consider a flower $F$ defined by some clique~$C$. The boundary 
$\bd F$ of $F$ is comprised of maximal pieces of the boundaries of 
the disks in $C$ that show up on~$\bd F$. We call these pieces 
\DEF{boundary arcs} and we denote the set of boundary arcs of a flower~$F$ by~$\B(F)$. For a set $\F$ of flowers, we define
$\B(\F) := \bigcup_{F\in\F} \B(F)$ to be the set of boundary arcs of the flowers in~$\F$.
\begin{lemma} \label{lem:flower-intersections}
Let $\F$ be a set of flowers that consist of $n$ disks in total.
Then the number of intersections between the boundary arcs in $\B(\F)$ is $O(n\cdot \ply(\F))$.
\end{lemma}
\begin{proof} 
First, observe that $|\B(F)|=O(n_F)$, where $n_F$ is the number of disks defining a flower~$F\in\F$, 
since the union of a set of
disks has linear complexity. Hence, $|B(\F)|=O(n)$.

Now consider a boundary arc $\beta\in\B(F)$ of some flower $F\in \F$. We assume 
for simplicity that the angle spanned by the circular arc $\beta$ is at most~$\pi/2$; 
if this is not the case we can split $\beta$ into at most four sub-arcs, 
and work with the sub-arcs instead. Let $D$ be the disk that contributes
the arc~$\beta$, let $c$ be the center of~$D$, and let $s_1$
and $s_2$ be the segments that connect $c$ to the endpoints of $\beta$.
Together with~$\beta$, these two segments bound a (convex) circular sector.
Let $Z_F$ be the set of such circular sectors created for the arcs in $\B(F)$; 
see Figure~\ref{fig:sectors}(i).
\begin{figure}
\begin{center}
\includegraphics[width=\textwidth]{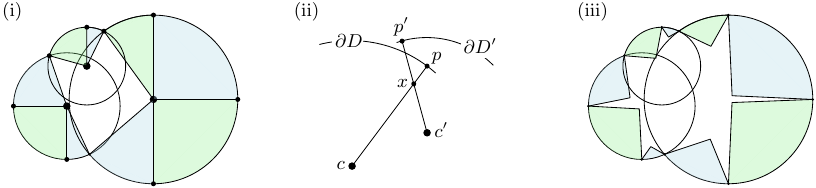}
\end{center}
\caption{(i) The set $Z_F$ of circular sectors defined by the boundary arcs of a flower~$F$. Some boundary arcs have been cut into sub-arcs to ensure they span an angle of at most~$\pi/2$. 
(ii)~If $|px|<|p'x|$ then $p\in D'$; otherwise $p'\in D$. 
(iii)~The modified sectors.} 
\label{fig:sectors}
\end{figure}
We say that two sectors $z,z'$ are \DEF{non-overlapping} if their interiors are disjoint.
\begin{claiminproof}
The sectors in the set $Z_F$ are pairwise non-overlapping.
\end{claiminproof}
\begin{proofinproof}
Consider two sectors $z,z'\in Z_F$. Let $D,D'$ be the disks whose boundaries
define $z$ and $z'$, respectively, and let $c$ and $c'$ be their centers.
If $D=D'$ then obviously $z$ and $z'$ do not overlap, so assume that $D\neq D'$. 
Since a sector in $Z_F$ cannot be fully contained in another sector in $Z_F$ 
by construction, the sectors $z,z'$ can only overlap if there is a proper 
intersection between their boundaries. Such an intersection can only happen between 
a segment~$s$ connecting $c$ to an endpoint~$p$ of some boundary arc $\beta\subset \bd D$
and a segment~$s'$ connecting $c'$ to an endpoint $p'$ of some boundary arc $\beta'\subset \bd D'$.
But then, by the triangular inequality, either $p\in D'$ or $p'\in D$---see 
Figure~\ref{fig:sectors}(ii) for an illustration---contradicting that $p$ and $p'$ 
are both points on $\bd F$.
\end{proofinproof}
Now consider a circular sector~$z$, defined by some arc~$\beta$. Let $x$ be
the apex of $z$ (which coincides with the center of the disk $D$ 
contributing $\beta$). If the angle at apex~$x$ is large enough then $z$ is
fat, but if the angle is very small then this is not the case. We therefore
proceed as follows: we move the apex~$x$ towards the midpoint of $\beta$ until
the angle at $x$ is slightly larger than~$\pi/2$; see Figure~\ref{fig:sectors}(iii). 
It is easy to see that the
modified region (which is no longer a circular sector) is fat and convex.
Moreover, no two modified regions share an apex anymore.
With a slight abuse of notation, from now on we use $Z_F$ to denote the set 
of modified regions created for a flower~$F\in\F$.

Let $Z(\F) := \bigcup_{F\in \F}Z_F$. Note that $|Z(\F)|\leq 4\cdot|\B(\F)|=O(n)$ and
that $\ply(Z(\F))\leq 2\cdot \ply(\F)=O(\log n)$; the latter is true since the regions
created for a single flower~$R$ are pairwise disjoint except at shared endpoints
of boundary arcs. 
For each arc~$\beta\in \B(\F)$, there is a region $z\in Z(\F)$ such that $\beta \subset \bd z$. 
Hence, the number of intersections between the arcs in $\B(\F)$ is upper bounded
by the number of intersections between the regions in~$Z(\F)$.
\begin{claiminproof}
The number of intersections between the regions in~$Z(\F)$ is $O(|Z(\F)| \cdot \ply(Z(\F)))$.
\end{claiminproof}
\begin{proofinproof}
The proof is standard, but we sketch it for completeness. We charge the intersection between two regions $z,z'$ to the smaller of the two
regions, and claim that each region is charged $O(\ply(Z(\F)))$ times. To see this, consider a region $z\in Z(\F)$ and let $b$ be
the smallest enclosing disk of~$Z$. Let $b'$ be the disk with the same center as~$b$ such that $\radius(b')=2\cdot\radius(b)$. Due to the fatness of the regions in $Z(\F)$, any region $z'$ that intersects $z$ and is at least as large as~$z$, will cover a constant
fraction of~the area of~$b'$. Hence, there can be at most $O(\ply(Z(\F)))$ such regions~$z'$.
\end{proofinproof}
This finishes the proof.
\end{proof}
For each clique~$C\in\C$, we can compute the corresponding flower~$F(C)$ 
in $O(|C|\log|C|)$ time using a simple divide-and-conquer algorithm.
Indeed, merging the flowers resulting
from two recursive calls can be done in linear time,
because the boundary arcs of a flower are sorted in circular order
around the stabbing point of the flower.
Since each disk is part of only one clique, this implies that the flower set~$\F$ 
can be computed in $O(n\log n)$ time. Moreover, 
the total number of boundary arcs in the set $\B(\F)$ of boundary arcs
is~$O(n)$ because the complexity of a single flower is linear.

Lemmas~\ref{lem:clique-construction} and~\ref{lem:flower-intersections} together
imply that the number of intersections in the set $\B(\F)$ of boundary arcs
is $O(n\log n)$. This immediately gives us an efficient approach to
find all pairs $F,F'\in\F$ of flowers whose \emph{boundaries} intersect: 
Simply compute all intersections in the set $\B(\F)$ of boundary
arcs, and for every pair of intersecting arcs $\beta\in \B(F)$ and $\beta'\in \B(F')$, 
add $(C,C')$ to~$E_{\cH}$ where $C,C'\in \C$ are the cliques defining the 
flowers~$F,F'$. These intersections can be computed in $O(n\log n + k)$ time,
where $k$ is the number of intersections, using a deterministic algorithm by Balaban~\cite{Balaban95}.
In our setting $k=O(n\log n)$, so the algorithm runs in $O(n\log n)$ time.
\medskip

It remains to compute the containment between flowers.
A natural approach for this is the following.
First, construct the arrangement~$\A(\F)$ induced by the flower set~$\F$.
Next, traverse the dual graph of~$\A(\F)$,
maintaining a list $\cL$ of flowers containing the face we are currently in. 
Whenever we enter a flower~$F$ for the first time during the traversal, 
we report the pairs $(F,F’)$ for the faces $F’$ that are currently in the list~$\cL$. 
This simple approach leads to the following result.
\begin{proposition} \label{prop:clique-cover-graph-expected}
Let $\D$ be a set of $n$ disks in the plane. Then we can compute a clique-based contraction
for $\ig[\D]$ in $O(n\log n)$ expected time.
\end{proposition} 
\begin{proof}
We have already discussed how to compute in $O(n\log n)$ time
the flower set $\F$ and the pairs of flowers whose boundary intersect.
It remains to compute the containment between flowers of $\F$.

We construct the arrangement~$\A(\F)$ in
$O(n\log n + k)$ expected\footnote{Balaban's deterministic algorithm for line-segment
intersection~\cite{Balaban95}, which we used to compute the intersections between
the flower boundaries, also works for curves. Unfortunately, that
algorithm does not seem to give enough information to construct the arrangement.
We cannot use the deterministic algorithm 
for line-segment intersection by Chazelle and Edelsbrunner~\cite{DBLP:journals/jacm/ChazelleE92} either,
because that algorithm does not work for curves.} time, 
where $k$ is the number of intersections between the boundary arcs,
using a randomized incremental algorithm by Mulmuley~\cite{DBLP:journals/jacm/Mulmuley91}.
Since $k=O(n\log n)$ by Lemmas~\ref{lem:clique-construction} 
and~\ref{lem:flower-intersections}, constructing the arrangement thus 
takes $O(n\log n)$ expected time. 
The traversal of the dual graph, including the maintenance of $\cL$ takes 
$O(n\log n)$ time. Since the size of 
the list $\cL$ is $O(\log n)$ at any time---this is because the ply of the 
flower set is $O(\log n)$---we report at most $O(n\log n)$ pairs in total
during the traversal. 
Thus, the total time to compute $E_{\cH}$ is $O(n\log n)$.
\end{proof}
Proposition~\ref{prop:clique-cover-graph-expected} gives a simple randomized 
$O(n\log n)$ algorithm to compute the clique-based contraction.
Randomization is used only to compute the edges $(C,C')$ 
for cliques $C,C'\in \C$ such that $F(C)$ is completely contained in $F(C')$,
or vice versa. 
We next explain a slightly more complicated deterministic approach for this.
We start with a lemma that allows us to determine which points 
are contained in the union of a set of disks (such as a flower).
\begin{lemma}\label{lem:disks+points}
Let $P$ be a set of $n$~points in the plane, and assume we are given the Voronoi 
diagram~$\vd(P)$ of~$P$. Let~$\D$ be a set of $m$~disks.
Then we can compute the subset of points from $P$ that are
contained in $\bigcup \D$ in $O(n+ m\log m)$ time.
\end{lemma}
\begin{proof}
We will transform the problem to computing the intersection of two $3$-dimensional
polyhedra, using the well-known lifting map and the correspondence between Voronoi
diagrams in~$\Reals^2$ and upper envelopes in~$\Reals^3$~\cite{GuibasS85,m-lecture-notes}.

Let $U: = \{(x,y,z)\in \REAL^3 : z=x^2+y^2\}$ be the unit paraboloid in~$\Reals^3$.
Let $\lambda:\Reals^2 \rightarrow U$ be the \emph{lifting map} that maps a point $q=(q_x,q_y)$ 
to the point $\lambda(q) := (q_x,q_y,q_x^2+q_y^2)$ on~$U$. 
It is well known that the image $\lambda(\bd D)$ 
of the boundary of a disk~$D$ is the intersection of a non-vertical plane 
with the unit paraboloid~$U$. We denote this plane by~$\pi_D$.
\factinproof{A}{Let $q\in\Reals^2$ be a point and let $D$ be a disk. Then
$q\not\in D$ iff $\lambda(q)$ lies strictly above~$\pi_D$.}
Let $\Pi_{\D} := \{ \pi_D : D\in D\}$ be the set of planes defined
by the disks in~$\D$, and let $\ue(\Pi_{\D})$ be their upper envelope.
Fact~A implies that $p\not\in \bigcup \D$ iff $\lambda(p)$ lies strictly above $\ue(\Pi_{\D})$.
We now explain how to test this efficiently for all points in~$P$.

Consider a point $p\in P$ and
let $\pi_p$ be the tangent plane of~$U$ at the point~$\lambda(p)$.
Let $\Pi_P := \{ \pi_p : p\in P\}$ and let $\ue(\Pi_P)$ denote the
upper envelope of~$\Pi_P$.  Every plane~$\pi_p\in\Pi_P$ contributes
a facet to $\ue(\Pi_P)$, which contains~$\lambda(p)$ in its interior.
Define $\U := \ue(\Pi_\D \cup \Pi_P)$ to be the upper envelope of $\Pi_\D \cup \Pi_P$. 
The crucial observation is the following.
\begin{claiminproof}
$p\not\in\bigcup \D$ iff $\lambda(p)$~lies 
in the relative interior of a facet of~$\U$. 
\end{claiminproof}
\begin{proofinproof}
As already noted, $\lambda(p)$ lies in the interior of a facet
of the upper envelope of $\Pi_P$. Thus, $\lambda(p)$~does not lie in the relative 
interior of a facet of~$\U$ if and only if some plane~$\pi_D\in \Pi_{\D}$ 
passes through or lies above~$\lambda(p)$. By Fact~A this is precisely 
the condition for~$D$ to contain~$p$.
\end{proofinproof}
We can thus determine which points $p\in P$ are contained in~$\bigcup \D$ as follows:
First, construct $\ue(\Pi_D)$ in $O(m\log m)$ time~\cite{DBLP:series/eatcs/Edelsbrunner87}.
Next, construct $\ue(\Pi_P)$. This can trivially be done in $O(n)$ time because we are
given~$\vd(P)$. Indeed, the vertical projection of $\ue(\Pi_P)$ onto the $xy$-plane
is equal to  $\vd(P)$---see for example~\cite[Section~8]{GuibasS85} 
or~\cite[Section 5.7]{m-lecture-notes}---so all we need to do is lift $\vd(P)$
to $\Reals^3$ in the appropriate manner. (More precisely, the Voronoi cell $\V(p)$
of a point~$p\in P$ is lifted to a facet of $\ue(\Pi_P)$ that projects onto $\V(p)$ and
that is contained in the plane~$\pi_p$.)
Since $\ue(\Pi_\D)$ and $\ue(\Pi_P)$ have complexity $O(m)$ and $O(n)$, respectively,
we can now compute $\U$ in $O(n+m)$ time using a linear-time algorithm
for intersecting two convex polyhedra~\cite{Chan16,Chazelle92}.
Finally, for each $p\in P$ we check if $p\in\bigcup \D$ as follows.
If $\pi_p$ does not contribute a facet to $\U$ then $p\in\bigcup \D$, and
if $\pi_p$ contributes a facet~$f_p$ to $\U$ then $p\not\in\bigcup \D$
iff $\lambda(p)$ lies in the interior of $f_p$. The latter test can be performed
in $O(|f_p|)$ time, where $|f_p|$ is the combinatorial complexity of~$f_p$. 
Since $\sum_{p\in P} |f_p|$ is bounded by the total complexity of $\U$,
which is $O(n+m)$, this takes $O(n+m)$ time in total.
\begin{figure}
\begin{center}
\includegraphics{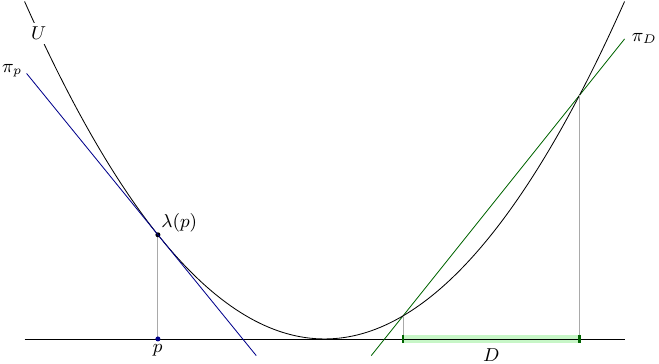}
\end{center} 
\caption{Two-dimensional illustration of the lifting map. Note that a point~$q$
lies outside~$D$ iff its lifted version~$\lambda(q)$ lies  strictly above~$\pi_D$.}
\label{fig:parabola}
\end{figure}
\end{proof}
We can now provide the deterministic algorithm to compute the edges 
of the clique-based contraction.
\begin{proposition} \label{prop:clique-cover-graph}
Let $\D$ be a set of $n$ disks in the plane. Then we can compute a clique-based contraction
for $\ig[\D]$ in $O(n\log n)$ time.
\end{proposition}
\begin{proof}
We can follow the proof of Proposition~\ref{prop:clique-cover-graph-expected}, 
except that we need to (deterministically) compute the pairs~$(F,F')\in \F \times \F$ 
such that~$F\subset F'$. To this end we pick an arbitrary point~$p_F$ inside each flower~$F\in \F$
and compute the pairs $(p_F,F')\in P_\F\times \F$ such that $p_F\in F'$, where
$P_\F := \{ p_F : F\in \F\}$. There are only $O(n\log n)$ such pairs because
$\ply(\F)=O(\log n)$ by Lemma~\ref{lem:clique-construction}.
Note that $p_F\in F'$ does not necessarily imply~$F \subset F'$;
however, it does imply $F\cap F'\neq \emptyset$ and therefore
we are still reporting only edges of the clique-based contraction.  
Recall that $\F=\bigcup_{\Gamma\in \Xi} \F(\Gamma)$. Hence, it suffices 
to separately compute, for each of the three hierarchical grids~$\Gamma\in \Xi$,
the pairs $(p_F,F')\in P_\F\times \F(\Gamma)$ such that $p_F\in F'$.

Fix a hierarchical grid $\Gamma\in \Xi$ and consider the corresponding
skip quadtree~$\tree(\Gamma)$. For a node~$\nu$ in~$\tree(\Gamma)$, 
let~$P_\nu$ be the subset of points from~$P_\F$ contained in the region~$R_\nu$.
Recall that~$\D_\nu$ is the set of disks from $\D(\Gamma)$ assigned to~$\nu$,
and that we partitioned $\D_\nu$ into a set $\C_\nu$ of~$O(1)$ stabbed cliques.
Let $\F_\nu=\{ F(C) : C\in \C_\nu\}$ be the flower set corresponding to these cliques.
As shown in the proof of Lemma~\ref{lem:clique-construction}, the assignment
of disks has the following property: if
a point~$q\in\Reals^2$ is contained in a disk~$D\in\D(\Gamma)$ then $D$ must be assigned
to a node on the root-to-leaf search path for~$q$ in~$\tree(\Gamma)$ or to a sibling of such a node. 
Hence, we can compute all pairs $(p_F,F')\in P_\F\times \F(\Gamma)$ such that $p_F\in F'$
with the following algorithm.
\begin{enumerate}
\item \label{step:det1} Construct the skip quadtree $\tree(\Gamma)$. 
\item \label{step:det2} Search with each point $p_F\in P_\F$ in $\tree(\Gamma)$, to construct
      the set $P_\nu$ for each node~$\nu$ in $\tree(\Gamma)$.  
\item \label{step:det3} Construct $\vd(P_\nu)$ for each node~$\nu$ in $\tree(\Gamma)$ in a bottom-up manner,
      so that when we construct $\vd(P_\nu)$ for an internal node~$\nu$
      we already have computed $\vd(P_\mu)$ for all children~$\mu$ of~$\nu$.  
\item \label{step:det4} Let $\sibl(\nu)$ be the set of siblings of a node~$\nu$ in $\tree(\Gamma)$,
      including $\nu$ itself. For each node~$\nu$ of~$\tree(\Gamma)$, each $\mu\in\sibl(\nu)$, 
      and each flower~$F'\in\F_\nu$, compute $P_\mu\cap F'$ using Lemma~\ref{lem:disks+points},
\end{enumerate}
The correctness of the algorithm follows from the discussion above. We now prove
that it runs in $O(n\log n)$ time.

Step~\ref{step:det1} can be done in $O(n\log n)$ time, as argued earlier.
Since $\tree(\Gamma)$ has depth $O(\log n)$, Step~\ref{step:det2} takes
$O(n \log n)$ time as well. Because each node~$\nu$ in $\tree(\Gamma)$ has $O(1)$
children and merging two Voronoi diagrams can be done 
in linear time~\cite{Chan16,Chazelle92,Kirkpatrick79}, Step~\ref{step:det3} 
takes $O(|P_\nu|)$ time at each node~$\nu$. Hence, it takes $O(n \log n)$ time in total. 
To analyze Step~\ref{step:det4}, observe that each flower~$F'\in\F_\nu$
is defined by at most $|\D_\nu|$~disks. Hence, Lemma~\ref{lem:disks+points}
tells us that we spend $O(|P_\mu|+ |\D_\nu| \log n)$ time
for each pair $\nu,\mu$ of siblings in~$\tree(\Gamma)$.
Since $|\F_\nu|=O(1)$ and each node of~$\tree(\Gamma)$
has $O(1)$~siblings, Step~\ref{step:det4} takes time
$\sum_{\nu\in \tree(\Gamma)} O(|P_\nu|+ |\D_\nu| \log n)$ time in total,
which is $O(n\log n)$ because  $\sum_{\nu\in \tree(\Gamma)} |P_\nu| = O(n\log n)$
and $\sum_{\nu\in \tree(\Gamma)} |D_\nu| = O(n)$.
\end{proof}

\subparagraph{Putting it all together.}
Proposition~\ref{prop:clique-cover-graph} gives us the first ingredient we need  
to apply Theorem~\ref{thm:framework}. The second ingredient is an algorithm that solves
\bit for disks. As observed in previous papers~\cite{ChanS19,Klost23},
this can be done in $O((n_B+n_R)\log n_R)$ time, as follows.
First, compute the additively weighted Voronoi diagram $\vd(R)$ 
on the centers of the set $R$ of red disks, where the weight of a center is
equal to the radius of its corresponding disk.
Next, query with the center of each blue disk $D$ in $\vd(R)$ to find
the red disk~$D'$ closest to~$D$. Now $D$ intersects $\bigcup R$, 
the union of the red disks, iff $D$ intersects~$D'$. Since $\vd(R)$
can be computed in $O(n_R \log n_R)$ time~\cite{Fortune87}, after which it can be
preprocessed for logarithmic-time point location in $O(n_R \log n_R)$ 
time~\cite{EdelsbrunnerGS86}, we can solve \bit in $O((n_B+n_R)\log n_R)$ time. We thus
obtain our main result.
\begin{theorem} \label{thm:sssp-for-disks}
Let $\D$ be a set of $n$ disks in the plane. Then we can compute a shortest-path tree 
in $\ig[\D]$ for a given source disk $\dsource\in \D$ in $O(n\log n)$ time.
\end{theorem}

\section{Extension to fat triangles}
\label{sec:fat}
We now show how to adapt our approach so that it can solve the SSSP problem 
on intersection graphs of fat triangles.
A triangle $\Delta$ is called \DEF{$\alpha$-fat} it its minimum angle is at least~$\alpha$.
Let $\D=\{\Delta_1,\ldots,\Delta_n\}$ be a set of $\alpha$-fat triangles, where
$\alpha>0$ is some fixed absolute constant. From now on, we simply refer to the triangles
in $\D$ as \DEF{fat triangles} and we refer $\alpha$ as the \DEF{fatness constant}.

\subparagraph{Constructing the clique-based contraction.}
Our approach to construct a clique-based contraction for fat triangles is similar to the 
algorithm for disks. We now go over the various ingredients and explain how to adapt them,
if necessary.

We first observe that Lemma~3.2 from Chan's paper~\cite{C-PTAS-fat} actually holds
for any type of objects. Hence, our Lemma~\ref{lem:shifted-quadtree} also holds
for any type of objects. To construct the skip quadtree, the set $P(\Gamma)$ of disk centers 
is replaced by a set $P(\Gamma)$ that contains an arbitrary point in each object. 
Constructing the cliques can therefore be done in exactly the same way as before. 
Indeed, the crucial property was as follows:
any set of disks intersecting the boundary of a region~$R_{\nu}$ and whose
size is at least~$\tfrac{1}{6} \cdot \size(R_{\nu})$, can be stabbed by~$O(1)$ points.
This property also holds for fat objects. 

Now let $C$ be a stabbed clique of fat triangles. We refer to the union of the
triangles in $C$ as a \DEF{spiky flower}, and we denote it by $F(C)$.
As before, we let $\C(\Gamma)$ denote the set of cliques created by our algorithm,
and we define $\F(\Gamma) := \{ F(C): C \in \C(\Gamma) \}$ to be the corresponding
set of spiky flowers. Then Lemma~\ref{lem:clique-construction} 
also holds for fat triangles---we can follow the proof verbatim, only
replacing occurrences of ``disk'' by ``fat triangle.''

It remains to prove the equivalent of Lemma~\ref{lem:flower-intersections}
for fat triangles. To keep the terminology similar to the case of disks, we
will refer to the edges of a spiky flower $F$ as \DEF{boundary segments},
and we denote the set of boundary segments of a flower~$F$ by~$\B(F)$.
Furthermore, we let $\F := \bigcup_{\Gamma\in \Xi} \F(\Gamma)$ denote the set of spiky flowers
created by our algorithm, and we define $\B(\F) := \bigcup_{F\in\F} \B(F)$.
\begin{lemma} \label{lem:spiky-flower-intersections}
Let $\F$ be a set of spiky flowers that consist of $n$ fat triangles in total.
Then the number of intersections between the boundary segments in $\B(\F)$ is $O(n\cdot \ply(\F))$.
\end{lemma}
\begin{proof}
We apply the same proof technique as in the proof of Lemma~\ref{lem:flower-intersections}:
we cover the boundary segments of each spiky flower~$F$ by a set $Z_F$ of 
non-overlapping fat ``sectors'' contained in the flower, as explained below,
from which the lemma follows.

The number of boundary segments of a spiky flower 
is~$O(n_F)$, where $n_F$ is the number of triangles defining~$F$~\cite[Section~3.2]{AgarwalKS95}.  
We create the set $Z_F$ for a spiky flower~$F$ as follows.
For each boundary segment~$\beta$ of~$F$, we create an isosceles triangle $z\subset F$ 
whose angles at the endpoints of~$\beta$ are~$\alpha/3$, where $\alpha$ is the fatness
constant of the triangles; see Figure~\ref{fig:triangle-sectors}. 
\begin{figure}
\begin{center}
\includegraphics{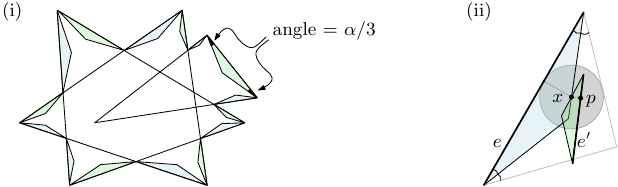}
\end{center}
\caption{(i) The set $Z_F$ of sectors defined by the boundary segments of a
spiky flower~$F$. 
(ii)~The point~$p$ lies in $D$, which in turn lies inside the expanded sector.}
\label{fig:triangle-sectors}
\end{figure}
With a slight abuse
of terminology, we refer to these isosceles triangles as sectors.
\begin{claiminproof}
The sectors in the set $Z_F$ are pairwise non-overlapping.
\end{claiminproof}
\begin{proofinproof}
Each sector~$z$ has one edge that corresponds to a boundary segment of~$F$; 
we call this the \DEF{external edge} of~$z$. The two other edges, which
make an angle $\alpha/3$ with the external edge, are called \DEF{internal edges}.

Now assume for a contradiction that two sectors $z,z'$ overlap.
Since the external edge of any sector cannot properly intersect an (external or internal)
edge of any other sector, there  must
be a crossing between an internal edge of $z$ and an internal edge of~$z'$.
Let $x$ be this crossing, and let $e$ and $e'$ be the external edges of~$z$ and $z'$, 
respectively; see Figure~\ref{fig:triangle-sectors}(ii). Assume without loss
of generality that the distance from $x$ to~$e$ is at least the distance
from $x$ to~$e'$. Let $p$ be the point on $e'$ nearest to~$x$, and
let~$D$ be the disk centered at~$x$ and touching~$e$. Then $p$ is contained in~$D$.
Moreover, $D$ is contained in the triangle~$\Delta\in \D$ contributing~$e$
to the boundary of the spiky flower~$F$. Indeed, if we were to expand the sector~$z$
by making the angles at the endpoints of~$e$ equal to $2\alpha/3$ then $D\subset z$,
and this expanded sector is contained in~$\Delta$.
Thus, $p\in \Delta$, which contradicts that $p$ lies on an external edge.
\end{proofinproof}
As in the proof of Lemma~\ref{lem:flower-intersections}, it now follows that
the number of intersections between the sectors in~$Z(\F)$ is $O(|Z(\F)|\cdot \ply(Z(\F)))$,
which proves the lemma.
\end{proof}
We can now compute the clique-based contraction with a deterministic algorithm that is
essentially the same as the randomized algorithm we used for disks: First, we compute  
each spiky flower $F$ of $\F$ using a simple divide-and-conquer approach;
because the combinatorial complexity of a spiky flower is linear, the whole
construction takes $O(n\log n)$ time. Next, we compute the intersection
between the boundary segments of $\B(\F)$ in $O(n\log n)$ time, using
the algorithm by Chazelle and Edelsbrunner~\cite{DBLP:journals/jacm/ChazelleE92}.
This algorithm can actually construct the arrangement~$\A(\F)$ as well---we
do not need Mulmuley's randomized algorithm~\cite{DBLP:journals/jacm/Mulmuley91}---so
that we can traverse the dual graph of $\A(\F)$ to 
find the pairs $(F,F')\in\F\times\F$ with $F\subset F'$.
We obtain the following result.
\begin{proposition} \label{prop:clique-cover-graph-triangles}
Let $\D$ be a set of $n$ fat triangles in the plane. Then we can compute a 
clique-based contraction
for $\ig[\D]$ in $O(n\log n)$ time.
\end{proposition}

\subparagraph{Bichromatic Intersection Testing for fat triangles.}
The final ingredient that we need is an efficient algorithm for \bit for 
a set $B$ of $n_B$ blue triangles and a set $R$ of $n_R$ red triangles, where
all triangles in $B\cup R$ are fat. We present an algorithm for
this problem that runs in $O(n \log^2 n)$ time.
\medskip

Let $A:=\{i\cdot (\alpha/2): 0\leq i \leq \lfloor 4\pi/\alpha \rfloor\}$ be a set 
of $O(1/\alpha)=O(1)$ \DEF{canonical directions}, where $\alpha$ is the fatness constant
of the triangles. Define a \DEF{canonical segment}
to be a segment whose direction is canonical, and define a \DEF{canonical chord} of
a fat triangle~$\Delta$ to be a canonical segment that connects a vertex of $\Delta$
to its opposite side; see Figure~\ref{fig:canonical}(i). 
\begin{figure}
\begin{center}
\includegraphics{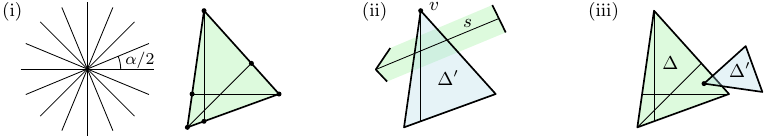}
\end{center}
\caption{(i) The set $A$ of canonical directions and the canonical chords of a triangle.
(ii)~Chord~$s$ intersects two sides of~$\Delta'$ so it intersects a canonical chord of~$\Delta'$.
(iii)~All canonical chords of~$\Delta$ miss $\Delta'$, so $\Delta'$ must have a vertex inside~$\Delta$.}
\label{fig:canonical}
\end{figure}
Note that each vertex of~$\Delta$ admits a canonical chord.
Let $S(\Delta)$ be a set of three canonical chords of $\Delta$,
one per vertex of $\Delta$,
and let $P(\Delta)$ be the endpoints of these chords
(including the vertices of~$\Delta$).
\begin{observation}\label{obs:canonical}
If a fat triangle $\Delta$ intersects a fat triangle $\Delta'$ then at least one
of the following conditions holds:
\begin{enumerate}[(i)]
\item $\Delta$ contains a vertex of $\Delta'$;
\item a point from $P(\Delta)$ lies inside~$\Delta'$;
\item a canonical chord from $S(\Delta)$ intersects a canonical chord from~$S(\Delta')$.
\end{enumerate}
\end{observation}
\begin{proof}
Suppose that condition~(ii) does not hold. Then there is a canonical chord~$s\in S(\Delta)$
that intersects two edges of~$\Delta'$, or all canonical chords in $S(\Delta)$ are disjoint 
from~$\Delta'$. In the former case, $s$~separates a vertex~$v$ of~$\Delta'$ from its
opposite side, which implies that $s$ intersects the canonical chord of~$v$; see Figure~\ref{fig:canonical}(ii).
Thus, condition~(iii) holds in this case. In the latter case, $\Delta'$ must have a vertex
inside~$\Delta$ and condition~(i) holds; see Figure~\ref{fig:canonical}(iii).
\end{proof}
The next lemmas provide the data structures we use to handle the
three cases in Observation~\ref{obs:canonical}. For condition~(i) we need a structure
for \emph{emptiness queries} with a fat query triangle~$\Delta$ in a set $P$ points in
the plane: decide if $P\cap \Delta\neq \emptyset$ and, if so, report a witness point $p\in P\cap \Delta$.
\begin{lemma} \label{lem:emptiness-query}
There exists a data structure for emptiness queries with fat query triangles on a set~$P$ 
of $n$ points in the plane that has $O(\log^2 n)$ 
query time, uses $O(n\log^2 n)$ storage, and can be built in $O(n\log^2 n)$ time. 
\end{lemma}
\begin{proof}
We call a triangle \emph{semi-canonical} if it has two \emph{canonical edges},
that is, two edges with a canonical direction. 
Any fat triangle~$\Delta$ can be partitioned into four semi-canonical triangles: first partition
$\Delta$ by a canonical chord from one of its vertices, and then partition the
resulting sub-triangles using canonical chords from the other two vertices. Hence, it suffices
to construct a data structure for semi-canonical query triangles. 
There are $O(1/\alpha^2)=O(1)$ different classes of semi-canonical triangles,
depending on the directions used by its canonical edges.
We construct a separate structure for each class, as described next.

Assume wlog that each query triangle~$\Delta$ in the class under consideration is of the form 
$[x_\Delta,\infty)\cap [y_\Delta,\infty) \cap h_\Delta$, where $h_{\Delta}$
is a negative half-plane (that is, a half-plane lying below its bounding line).
We can answer emptiness queries with such query triangles using a three-level
data structure: the first level is a search tree on the $x$-coordinates
of the points, the second level a search tree on the $y$-coordinates, and the
third level a data structure for half-plane emptiness queries.

Emptiness query with a negative half-plane $h_\Delta$ on a set $P'$ of points
can be answered by binary search on the
lower hull $\lh(P')$ of~$P'$. Indeed, if $p'\in P'$ is a vertex of $\lh(P')$ that
admits a tangent line parallel to~$\bd h_\Delta$, the bounding line of~$h_\Delta$,
then $P'$ contains a point below~$\bd h_\Delta$ iff $p'$ lies below~$\bd h_\Delta$. 
See Figure~\ref{fig:DS-fat-1}(i).
Thus our third-level structures are sorted arrays of the relevant lower hulls.
Since $\lh(P')$ can be computed 
in $O(|P'|)$ time if $P'$ is presorted on $x$-coordinate,
our three-level data structures uses $O(n\log^2 n)$
storage and can be built in $O(n\log^2 n)$ time~\cite{DBLP:journals/jacm/WillardL85}.

\begin{figure}
\begin{center}
\includegraphics{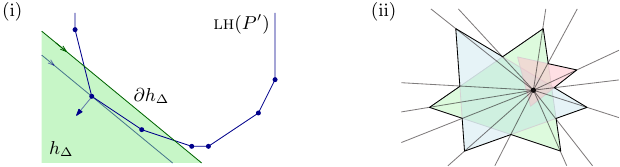}
\end{center}
\caption{(i)~The lower half-plane~$h_\Delta$ contains a point of~$P'$ if and only if
it contains the extreme vertex of~$\lh(P')$ in the direction orthogonal to~$\bd h_\Delta$.
(ii)~Partition of the plane into cones, induced by a spiky flower~$F$ defined by three triangles.}
\label{fig:DS-fat-1}
\end{figure}

A query in this three-level data structure takes $O(\log^3 n)$ time. 
Since the search in the third-level structures is just a binary search in a sorted list
and these searches are all with the same value---the direction of~$\bd h_\Delta$---we
can speed this up to $O(\log^2 n)$ using fractional 
cascading~\cite{DBLP:journals/algorithmica/ChazelleG86}, without increasing
the (asymptotic) preprocessing time.
\end{proof}
\emph{Note:} The data structure presented above is almost the same as the one by 
Gray~\cite[page~14]{gray-thesis} for reporting queries with fat query triangles.
The difference lies in the third level, where we can be more efficient because
we only want to answer emptiness queries.
\medskip

To handle condition~(ii) from Observation~\ref{obs:canonical} we need a data structure
for \emph{stabbing queries} on a set~$R$ of $n$ fat triangles: given a query point~$q$,
decide whether there is a triangle $~\Delta\in R$ such that~$q\in\Delta$ and, if so, report a witness triangle.
\begin{lemma}\label{lem:stabbing-query}
There exists a data structure for stabbing queries on a set~$R$ of $n$ fat triangles 
that has $O(\log^2 n)$ query time, uses $O(n)$ storage,
and can be built in $O(n\log n)$ time.
\end{lemma}
\begin{proof}
The data structure is based on our algorithm to construct a clique-based contraction, which
we run on the set~$R$. Thus, we consider the collection~$\Xi$ 
of three hierarchical grids from Lemma~\ref{lem:shifted-quadtree}, and we partition $R$ into 
three subsets $R(\Gamma)$ according to the hierarchical grid $\Gamma\in\Xi$ they are aligned with.
We then construct a skip quadtree~$\tree$ on the reference points of the triangles in~$R(\Gamma)$,
based on the hierarchical grid~$\Gamma$. We assign each triangle $\Delta\in\R(\Gamma)$
to a node in~$\tree$ as described earlier, and we construct~$O(1)$ stabbed cliques for each node.

Each of the resulting cliques~$C$ is preprocessed for stabbing queries, as follows.
Let $p\in\bigcap_{\Delta\in C} \Delta$, and let $F(C)$ be the
spiky flower defined by~$C$. We partition the plane into cones with apex~$p$ by drawing
rays from $p$ through every vertex of~$F(C)$. 
See Figure~\ref{fig:DS-fat-1}(ii).
Given a query point~$q$, we can now
answer a stabbing query in logarithmic time by binary search. 

The data structure for each of the three hierarchical grids
uses $O(n)$ storage for the skip quadtree
plus $\sum_C O(|C|)= O(n)$ for the cone structures, so $O(n)$ in total.
Since the skip quadtree can be built in $O(n\log n)$ time and each
spiky flower $F(C)$ with its corresponding cone structure can be built
in $O(|C|\log|C|)$ time, the total preprocessing time is~$O(n\log n)$.
To answer a query with a point~$q$ we search in $\tree$, and for
each node~$\nu$ on the search path we query the $O(1)$ cone structures
associated to~$\nu$. Thus the query time is~$O(\log^2 n)$.
\end{proof}
To handle condition~(iii) we need a structure for segment-intersection queries on a set
$S$ of canonical segments in the plane: given a canonical query segment~$s$,
decide $s$ intersects a segment $s'\in S$, and, if so, report a witness segment.
Obtaining an efficient structure for this is straightforward, as shown in the following lemma.
\begin{lemma}\label{lem:segment-query}
There exists a data structure for segment-intersection queries with canonical query segmwnts
in a set $S$ of $n$ canonical segments
that has $O(\log n)$ query time, uses $O(n)$ storage, and can be built in $O(n\log n)$ time.
\end{lemma}
\begin{proof}
We build $O(1/\alpha)=O(1)$ data structures
on~$S$, one for each of the possible directions of the query segment. 
To construct a data structure for a fixed direction of the query segment,
we partition $S$ into $O(1/\alpha)$ classes, one for each canonical direction, 
and we build a separate substructure for each class. 
Assume wlog that the query chord is vertical and that the chords in
in the given class are horizontal. Then we can use an interval tree~\cite[Section~10.1]{bcko-cgaa-08}
to answer queries in $O(\log n)$, using $O(n)$ space and $O(n\log n)$
preprocessing. 
\end{proof}
Putting everything together, we conclude that we can solve
\bit for fat triangles in $O((n_B+n_R) \log^2 (n_B+n_R))$ time:
we construct the data structure from Lemma~\ref{lem:emptiness-query}
on the vertices of the triangles in~$R$ and query with each triangle in~$B$,
we construct the data structure from Lemma~\ref{lem:stabbing-query}
on the triangles in~$R$ and query with each vertex of a triangle in~$B$,
and we construct the data structure from Lemma~\ref{lem:segment-query}
on the canonical chords of the triangles in~$R$ and query with the canonical
chords of each triangle in~$B$. This leads to the following theorem.
\begin{theorem} \label{thm:sssp-for-fat}
Let $\D$ be a set of $n$ fat triangles in the plane. Then we can compute a shortest-path tree 
in $\ig[\D]$ for a given source triangle $\Delta_{\mathrm{src}}\in \D$ in $O(n\log^2 n)$ time.
\end{theorem}

\section{Concluding remarks}
\label{sec:conclusions}
We presented the first $O(n \log n)$ algorithm for the SSSP problem in disk graphs.
We extended the algorithm to intersection graphs of fat triangles, where we obtain a running time of $O(n \log^2 n)$. 
The algorithm can be extended to the problem of computing a shortest-path tree
from a set $S$ of multiple sources, where the goal is to compute for each disk a shortest
path to the set~$S$. 

A natural question is whether the algorithm for fat triangles can be
improved further. Another natural question is whether an $O(n \polylog n)$ algorithm is 
possible for intersection graphs of non-fat objects such as line segments. 
This is unlikely, however, because Hopcroft's problem---given
a set $L$ of $n$ lines and a set $P$ of $n$ points, decide if any
of the points lies on any of the lines---can 
be reduced to the SSSP problem for segments, and Hopcroft's problem
has an $\Omega(n^{4/3})$ lower bound~\cite{DBLP:journals/dcg/Erickson96}
in a somewhat restricted model of computation.

\bibliography{references}

\begin{thebibliography}{10}

\bibitem{AgarwalKS95}
Pankaj~K. Agarwal, Matthew~J. Katz, and Micha Sharir.
\newblock Computing depth orders for fat objects and related problems.
\newblock {\em Comput. Geom.}, 5:187--206, 1995.
\newblock \href {https://doi.org/10.1016/0925-7721(95)00005-8}
  {\path{doi:10.1016/0925-7721(95)00005-8}}.

\bibitem{AronovBES14}
Boris Aronov, Mark de~Berg, Esther Ezra, and Micha Sharir.
\newblock Improved bounds for the union of locally fat objects in the plane.
\newblock {\em {SIAM} J. Comput.}, 43(2):543--572, 2014.
\newblock \href {https://doi.org/10.1137/120891241}
  {\path{doi:10.1137/120891241}}.

\bibitem{Balaban95}
Ivan~J. Balaban.
\newblock An optimal algorithm for finding segments intersections.
\newblock In {\em Proc.~11th Symposium on Computational Geometry ({SoCG})},
  pages 211--219, 1995.
\newblock \href {https://doi.org/10.1145/220279.220302}
  {\path{doi:10.1145/220279.220302}}.

\bibitem{CabelloJ15}
Sergio Cabello and Miha Jej{\v c}i{\v c}.
\newblock Shortest paths in intersection graphs of unit disks.
\newblock {\em Comput. Geom.}, 48(4):360--367, 2015.
\newblock \href {https://doi.org/10.1016/j.comgeo.2014.12.003}
  {\path{doi:10.1016/j.comgeo.2014.12.003}}.

\bibitem{C-PTAS-fat}
Timothy~M. Chan.
\newblock Polynomial-time approximation schemes for packing and piercing fat
  objects.
\newblock {\em J. Algorithms}, 46(2):178--189, 2003.
\newblock \href {https://doi.org/10.1016/S0196-6774(02)00294-8}
  {\path{doi:10.1016/S0196-6774(02)00294-8}}.

\bibitem{Chan10}
Timothy~M. Chan.
\newblock A dynamic data structure for 3-d convex hulls and 2-d nearest
  neighbor queries.
\newblock {\em J. {ACM}}, 57(3):16:1--16:15, 2010.
\newblock \href {https://doi.org/10.1145/1706591.1706596}
  {\path{doi:10.1145/1706591.1706596}}.

\bibitem{Chan16}
Timothy~M. Chan.
\newblock A simpler linear-time algorithm for intersecting two convex polyhedra
  in three dimensions.
\newblock {\em Discret. Comput. Geom.}, 56(4):860--865, 2016.
\newblock \href {https://doi.org/10.1007/S00454-016-9785-3}
  {\path{doi:10.1007/S00454-016-9785-3}}.

\bibitem{ChanH23}
Timothy~M. Chan and Zhengcheng Huang.
\newblock Constant-hop spanners for more geometric intersection graphs, with
  even smaller size.
\newblock In {\em Proc.~39th International Symposium on Computational Geometry
  (SoCG)}, volume 258 of {\em LIPIcs}, pages 23:1--23:16, 2023.
\newblock \href {https://doi.org/10.4230/LIPICS.SoCG.2023.23}
  {\path{doi:10.4230/LIPICS.SoCG.2023.23}}.

\bibitem{ChanS16}
Timothy~M. Chan and Dimitrios Skrepetos.
\newblock All-pairs shortest paths in unit-disk graphs in slightly subquadratic
  time.
\newblock In Seok{-}Hee Hong, editor, {\em Proc.~27th International Symposium
  on Algorithms and Computation (ISAAC)}, volume~64 of {\em LIPIcs}, pages
  24:1--24:13, 2016.
\newblock \href {https://doi.org/10.4230/LIPICS.ISAAC.2016.24}
  {\path{doi:10.4230/LIPICS.ISAAC.2016.24}}.

\bibitem{ChanS19}
Timothy~M. Chan and Dimitrios Skrepetos.
\newblock All-pairs shortest paths in geometric intersection graphs.
\newblock {\em J. Comput. Geom.}, 10(1):27--41, 2019.
\newblock \href {https://doi.org/10.20382/JOCG.V10I1A2}
  {\path{doi:10.20382/JOCG.V10I1A2}}.

\bibitem{Chazelle92}
Bernard Chazelle.
\newblock An optimal algorithm for intersecting three-dimensional convex
  polyhedra.
\newblock {\em {SIAM} J. Comput.}, 21(4):671--696, 1992.
\newblock \href {https://doi.org/10.1137/0221041} {\path{doi:10.1137/0221041}}.

\bibitem{DBLP:journals/jacm/ChazelleE92}
Bernard Chazelle and Herbert Edelsbrunner.
\newblock An optimal algorithm for intersecting line segments in the plane.
\newblock {\em J. {ACM}}, 39(1):1--54, 1992.
\newblock \href {https://doi.org/10.1145/147508.147511}
  {\path{doi:10.1145/147508.147511}}.

\bibitem{DBLP:journals/algorithmica/ChazelleG86}
Bernard Chazelle and Leonidas~J. Guibas.
\newblock Fractional cascading: I. {A} data structuring technique.
\newblock {\em Algorithmica}, 1(2):133--162, 1986.
\newblock \href {https://doi.org/10.1007/BF01840440}
  {\path{doi:10.1007/BF01840440}}.

\bibitem{CLRS09}
Thomas~H. Cormen, Charles~E. Leiserson, Ronald~L. Rivest, and Clifford Stein.
\newblock {\em Introduction to Algorithms}.
\newblock The MIT Press, 3rd edition, 2009.

\bibitem{bcko-cgaa-08}
Mark de~Berg, Otfried Cheong, Marc~J. van Kreveld, and Mark~H. Overmars.
\newblock {\em Computational Geometry: Algorithms and Applications (3rd
  Edition)}.
\newblock Springer, 2008.
\newblock \href {https://doi.org/10.1007/978-3-540-77974-2}
  {\path{doi:10.1007/978-3-540-77974-2}}.

\bibitem{DBLP:series/eatcs/Edelsbrunner87}
Herbert Edelsbrunner.
\newblock {\em Algorithms in Combinatorial Geometry}, volume~10 of {\em {EATCS}
  Monographs on Theoretical Computer Science}.
\newblock Springer, 1987.
\newblock \href {https://doi.org/10.1007/978-3-642-61568-9}
  {\path{doi:10.1007/978-3-642-61568-9}}.

\bibitem{EdelsbrunnerGS86}
Herbert Edelsbrunner, Leonidas~J. Guibas, and Jorge Stolfi.
\newblock Optimal point location in a monotone subdivision.
\newblock {\em {SIAM} J. Comput.}, 15(2):317--340, 1986.
\newblock \href {https://doi.org/10.1137/0215023} {\path{doi:10.1137/0215023}}.

\bibitem{EfratIK01}
Alon Efrat, Alon Itai, and Matthew~J. Katz.
\newblock Geometry helps in bottleneck matching and related problems.
\newblock {\em Algorithmica}, 31(1):1--28, 2001.
\newblock \href {https://doi.org/10.1007/S00453-001-0016-8}
  {\path{doi:10.1007/S00453-001-0016-8}}.

\bibitem{EppsteinGS08}
David Eppstein, Michael~T. Goodrich, and Jonathan~Z. Sun.
\newblock Skip quadtrees: Dynamic data structures for multidimensional point
  sets.
\newblock {\em Int. J. Comput. Geom. Appl.}, 18:131--160, 2008.
\newblock \href {https://doi.org/10.1142/S0218195908002568}
  {\path{doi:10.1142/S0218195908002568}}.

\bibitem{DBLP:journals/dcg/Erickson96}
Jeff Erickson.
\newblock New lower bounds for {H}opcroft's problem.
\newblock {\em Discret. Comput. Geom.}, 16(4):389--418, 1996.
\newblock \href {https://doi.org/10.1007/BF02712875}
  {\path{doi:10.1007/BF02712875}}.

\bibitem{Fortune87}
Steven Fortune.
\newblock A sweepline algorithm for {V}oronoi diagrams.
\newblock {\em Algorithmica}, 2:153--174, 1987.
\newblock \href {https://doi.org/10.1007/BF01840357}
  {\path{doi:10.1007/BF01840357}}.

\bibitem{gray-thesis}
Chris Gray.
\newblock {\em Algorithms for Fat Objects: Decompositions and Applications}.
\newblock PhD thesis, TU Eindhoven, 2008.
\newblock URL:
  \url{https://pure.tue.nl/ws/portalfiles/portal/3256691/200811208.pdf}.

\bibitem{GuibasS85}
Leonidas~J. Guibas and Jorge Stolfi.
\newblock Primitives for the manipulation of general subdivisions and
  computation of {V}oronoi diagrams.
\newblock {\em {ACM} Trans. Graph.}, 4(2):74--123, 1985.
\newblock \href {https://doi.org/10.1145/282918.282923}
  {\path{doi:10.1145/282918.282923}}.

\bibitem{HarPeledBook}
Sariel Har-Peled.
\newblock {\em Geometric approximation algorithms}, volume 173 of {\em
  Mathematical Surveys and Monographs}.
\newblock American Mathematical Society, 2011.
\newblock \href {https://doi.org/10.1090/surv/173}
  {\path{doi:10.1090/surv/173}}.

\bibitem{KaplanMRSS20}
Haim Kaplan, Wolfgang Mulzer, Liam Roditty, Paul Seiferth, and Micha Sharir.
\newblock Dynamic planar {V}oronoi diagrams for general distance functions and
  their algorithmic applications.
\newblock {\em Discret. Comput. Geom.}, 64(3):838--904, 2020.
\newblock \href {https://doi.org/10.1007/S00454-020-00243-7}
  {\path{doi:10.1007/S00454-020-00243-7}}.

\bibitem{Kirkpatrick79}
David~G. Kirkpatrick.
\newblock Efficient computation of continuous skeletons.
\newblock In {\em Proc.~20th Annual Symposium on Foundations of Computer
  Science (FOCS)}, pages 18--27. {IEEE} Computer Society, 1979.
\newblock \href {https://doi.org/10.1109/SFCS.1979.15}
  {\path{doi:10.1109/SFCS.1979.15}}.

\bibitem{Klost23}
Katharina Klost.
\newblock An algorithmic framework for the single source shortest path problem
  with applications to disk graphs.
\newblock {\em Comput. Geom.}, 111:101979, 2023.
\newblock \href {https://doi.org/10.1016/j.comgeo.2022.101979}
  {\path{doi:10.1016/j.comgeo.2022.101979}}.

\bibitem{doi:10.1137/20M1388371}
Chih-Hung Liu.
\newblock Nearly optimal planar $k$ nearest neighbors queries under general
  distance functions.
\newblock {\em SIAM Journal on Computing}, 51(3):723--765, 2022.
\newblock \href {https://doi.org/10.1137/20M1388371}
  {\path{doi:10.1137/20M1388371}}.

\bibitem{m-lecture-notes}
Jir{\'{\i}} Matousek.
\newblock {\em Lectures on Discrete Geometry}, volume 212 of {\em Graduate
  Texts in Mathematics}.
\newblock Springer, 2002.

\bibitem{DBLP:journals/jacm/Mulmuley91}
Ketan Mulmuley.
\newblock A fast planar partition algorithm, {II}.
\newblock {\em J. {ACM}}, 38(1):74--103, 1991.
\newblock \href {https://doi.org/10.1145/102782.102785}
  {\path{doi:10.1145/102782.102785}}.

\bibitem{RodittyS11}
Liam Roditty and Michael Segal.
\newblock On bounded leg shortest paths problems.
\newblock {\em Algorithmica}, 59(4):583--600, 2011.
\newblock \href {https://doi.org/10.1007/S00453-009-9322-3}
  {\path{doi:10.1007/S00453-009-9322-3}}.

\bibitem{DBLP:journals/jacm/WillardL85}
Dan~E. Willard and George~S. Lueker.
\newblock Adding range restriction capability to dynamic data structures.
\newblock {\em J. {ACM}}, 32(3):597--617, 1985.
\newblock \href {https://doi.org/10.1145/3828.3839}
  {\path{doi:10.1145/3828.3839}}.

\end{thebibliography}
\end{document}